\documentclass[10pt,fullpage]{article}
\usepackage{amsmath,amssymb,amsfonts,amsthm,epsfig}
\usepackage{cancel}

\usepackage[usenames,dvipsnames]{xcolor}
\usepackage{bm,xspace}
\usepackage{tcolorbox}
\usepackage{cancel}
\usepackage{fullpage}
\usepackage{pras}
\usepackage{framed}
\usepackage{verbatim}
\usepackage{enumitem}
\usepackage{array}
\usepackage{multirow}
\usepackage{afterpage}
\usepackage{mathrsfs}
\usepackage{pifont} 
\usepackage{chngpage}
\usepackage[normalem]{ulem}
\usepackage{boxedminipage}
\usepackage{caption}
\usepackage{tikz}

\def\colorful{1}

\ifnum\colorful=1

\newcommand{\red}[1]{{\color{red} {#1}}}

\fi
\ifnum\colorful=0

\newcommand{\red}[1]{{{#1}}}

\fi

\DeclareMathOperator{\plu}{plu}
\DeclareMathOperator{\SC}{SC}
\newcommand{\ones}{\textbf{\textnormal{1}}}
\newcommand{\cg}{\succ}
\newcommand{\cl}{\prec}
\newcommand{\cgeq}{\succeq}

\newcommand{\vect}[1]{\boldsymbol{#1}}
\renewcommand{\hat}[1]{\widehat{#1}}

\makeatletter
\newtheorem*{rep@theorem}{\rep@title}
\newcommand{\newreptheorem}[2]{
\newenvironment{rep#1}[1]{
 \def\rep@title{#2 \ref{##1}}
 \begin{rep@theorem}\itshape}
 {\end{rep@theorem}}}
\makeatother
\theoremstyle{plain}

\newreptheorem{theorem}{Theorem}

\makeatletter
\newtheorem*{rep@claim}{\rep@title}
\newcommand{\newrepclaim}[2]{
\newenvironment{rep#1}[1]{
 \def\rep@title{#2 \ref{##1}}
 \begin{rep@claim}\itshape}
 {\end{rep@claim}}}
\makeatother
\theoremstyle{plain}

\newrepclaim{claim}{Claim}

\begin{document}

\title{
 Metric Distortion Bounds for Randomized Social Choice 
 \vspace{15pt}}

\author{
Moses Charikar \\ \hspace{0pt}{Stanford University}  \\ \hspace{0pt}{\texttt{moses@cs.stanford.edu}}\thanks{Moses Charikar was supported by a Simons Investigator Award.} \and 
Prasanna Ramakrishnan \\ \hspace{0pt}{Stanford University} \\ \hspace{0pt}{\texttt{pras1712@stanford.edu}\thanks{Prasanna Ramakrishnan was supported by NSF CAREER Award 1942123.}}
}  

\date{\vspace{15pt}\small{\today}}

\maketitle

\begin{abstract} 
Consider the following social choice problem. Suppose we have a set of $n$ voters and $m$ candidates that lie in a metric space. The goal is to design a mechanism to choose a candidate whose average distance to the voters is as small as possible. However, the mechanism does not get direct access to the metric space. Instead, it gets each voter's ordinal ranking of the candidates by distance. Given only this partial information, what is the smallest worst-case approximation ratio (known as the \emph{distortion}) that a mechanism can guarantee?

A simple example shows that no deterministic mechanism can guarantee distortion better than $3$, and no randomized mechanism can guarantee distortion better than $2$. It has been conjectured that both of these lower bounds are optimal, and recently, Gkatzelis, Halpern, and Shah proved this conjecture for deterministic mechanisms. We disprove the conjecture for randomized mechanisms for $m \geq 3$ by constructing elections for which no randomized mechanism can guarantee distortion better than $2.0261$ for $m = 3$, $2.0496$ for $m = 4$, up to $2.1126$ as $m \to \infty$. We obtain our lower bounds by identifying a class of simple metrics that appear to capture much of the hardness of the problem, and we show that any randomized mechanism must have high distortion on one of these metrics. We provide a nearly matching upper bound for this restricted class of metrics as well. Finally, we conjecture that these bounds give the optimal distortion for every $m$, and provide a proof for $m = 3$, thereby resolving that case.

\end{abstract}

 \thispagestyle{empty}
\newpage





\section{Introduction}

Social choice theory studies ways in which a group of agents (or voters) can combine their preferences to make a collective choice among a number of alternatives (or candidates) in a way that maximizes their utility. This broad formulation can encompass a number of different problems by changing how the utility is measured, or how the preferences are compiled. One natural version of this problem that has garnered recent interest considers utility functions given by the distances in an underlying metric space, and preferences given in the form of ordinal rankings of the candidates in order of distance.

More precisely, suppose we have a set $V$ of $n$ voters and a set $C$ of $m$ candidates that lie in a metric space $(V\cup C, d)$. We would like to design a mechanism to choose the best candidate for the voters. Specifically, we define the \emph{social cost} of a candidate $c \in C$ to be their average distance to the voters:
$$\SC(c) := \frac1n \sum_{v \in V} d(c, v).$$
Under this measure, the lower a candidate's cost, the more preferable they are. If the mechanism had access to the metric directly, it could easily find the best candidate by computing $\SC(c)$ for each $c \in C$, and choosing $c^* = \arg\min_{c \in C} \SC(c).$ Instead, we suppose that the mechanism only has some partial information about the metric space, such as each voter's ranking of the candidates in order of distance. In this case, the mechanism's performance on an election instance is quantified by its approximation ratio, known in the literature as its \emph{distortion}. If the mechanism chooses a candidate $\hat{c} \in C$ (perhaps in a randomized fashion), then this is given by 
$$\frac{\E[\SC(\hat{c})]}{\min_{c \in C} \SC(c)}.$$
In general, we would like to bound the \emph{worst-case} distortion. For a fixed mechanism, this is the maximum distortion it incurs on any election instance.

This metric distortion setting has some natural motivations in practical social choice scenarios. For example, in choosing an electoral candidate in a democracy, the distance metric could represent a sort of ``ideological distance'' that aggregates agreement or disagreement about various policies. The triangle inequality then encodes the intuition that if a voter is close to two candidates, then those two candidates cannot be too far from each other. The metric space itself may be too complex or abstract to determine with sufficient accuracy, so the best we can do is elicit some information about the voters' preferences, such as their favorite candidate, or their ordinal ranking of the candidates. Of course, some of the theoretical considerations we have may not translate well to reality. For example, democracies may be apprehensive about involving randomness in their elections, and single-issue voters may disrupt the metric intuition. Nonetheless this theoretical framework can help us understand which kinds of information (that can be elicited from individual voters) are sufficiently expressive to allow a mechanism to approximately understand more nuanced aspects of the underlying preference space. 

Another application is in the facility location problem \cite{CG99,JMMSV03,CP19}, where the voters correspond to customers and the candidates correspond to possible locations for a new facility (e.g., a grocery store serving the customers in an area). Here, the metric space is much more concrete, but for the mechanism to have full access to the metric space would be to infringe on the privacy of the customers. Once again, it is natural to instead ask the customers for their preferences among the new facilities, which does not reveal too much about each individual customer, and may be enough information to choose a sufficiently good facility.

The following simple example, originally considered by \cite{ABP15}, shows that no deterministic mechanism can guarantee distortion less than 3, and no randomized mechanism can guarantee distortion less than 2. Suppose we have two candidates labeled $\{1, 2\}$, and half of the voters rank $(1, 2)$ (they prefer candidate $1$ over candidate $2$), and the other half rank $(2, 1)$. Consider the two possibilities for the underlying metric space given in Figure~\ref{fig:two-metrics}.

\begin{figure}[h!]
\centering
\begin{tikzpicture}[scale=1.5]
\node[above] at (1, 1.5) {Metric $d_1$:};
\draw[thick, color=green] (0, 1) -- (2, 1);
\node[above, color=green] at (1, 1) {0};

\draw[thick, color=orange] (0, 0) -- (2, 1);
\node[below, color=orange] at (1, .5) {2};

\draw[thick, color=blue] (0, 1) -- (2, 0);
\node[above, color=blue] at (1, .5) {1};

\draw[thick, color=blue] (0, 0) -- (2, 0);
\node[below, color=blue] at (1, 0) {1};
\draw[color=black, fill=black] (0, 0) circle (1pt);
\node[left] at (0, 0) {2};
\draw[color=black, fill=black] (0, 1) circle (1pt);
\node[left] at (0, 1) {1};
\draw[color=black, fill=black] (2, 0) circle (1pt);
\node[right] at (2, 0) {(2, 1)};
\draw[color=black, fill=black] (2, 1) circle (1pt);
\node[right] at (2, 1) {(1, 2)};
\end{tikzpicture} \qquad \qquad
\begin{tikzpicture}[scale=1.5]
\node[above] at (1, 1.5) {Metric $d_2$:};
\draw[thick, color=blue] (0, 1) -- (2, 1);
\node[above, color=blue] at (1, 1) {1};

\draw[thick, color=blue] (0, 0) -- (2, 1);
\node[below, color=blue] at (1, .5) {1};

\draw[thick, color=orange] (0, 1) -- (2, 0);
\node[above, color=orange] at (1, .5) {2};

\draw[thick, color=green] (0, 0) -- (2, 0);
\node[below, color=green] at (1, 0) {0};
\draw[color=black, fill=black] (0, 0) circle (1pt);
\node[left] at (0, 0) {2};
\draw[color=black, fill=black] (0, 1) circle (1pt);
\node[left] at (0, 1) {1};
\draw[color=black, fill=black] (2, 0) circle (1pt);
\node[right] at (2, 0) {(2, 1)};
\draw[color=black, fill=black] (2, 1) circle (1pt);
\node[right] at (2, 1) {(1, 2)};
\end{tikzpicture}
\caption{Two possible metric spaces for the election. The candidates are on the left, and the voters are on the right (lumped together by their preferences). The distance between a candidate and a voter is given by the label on the edge between them.}\label{fig:two-metrics}
\end{figure}
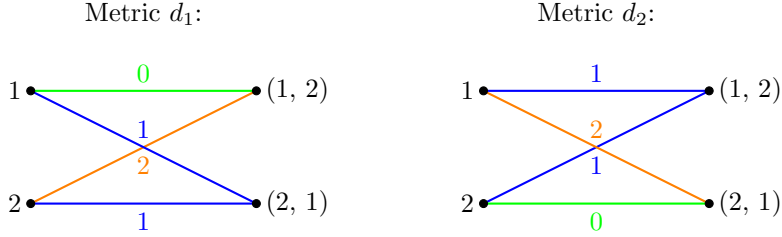

For $d_1$, we have $\SC(1) = \frac12$ and $\SC(2) = \frac32$, and for the metric space $d_2$ we have the opposite ( $\SC(1) = \frac32$ and $\SC(2) = \frac12$). Clearly then, any deterministic mechanism that picks candidate 1 incurs distortion 3 on $d_2$, and vice versa. Similarly, if a randomized mechanism picks candidates 1 and 2 with probability $p_1$ and $p_2$ respectively, then the worst-case distortion for these two metrics is 
$$\frac{\max(\frac32 p_1 + \frac12 p_2, \frac12 p_1 + \frac32 p_2)}{\frac12} = 2\left(\frac12 + \max(p_1, p_2)\right) \geq 2.$$
Note that these metric spaces can be realized in a simpler way, on the number line. The candidates 1 and 2 sit at the points 0 and 2, and either the $(1, 2)$ voters sit at 0 and the $(2, 1)$ voters sit at 1, or the $(1, 2)$ voters sit at 1 and the $(2, 1)$ voters sit at 2 (i.e., each group is either impartial, or has a strong preference). We show the above visualization since it foreshadows the metric spaces we will use later on.

A major conjecture in this framework is that the lower bounds given by the above example are optimal (see \cite[Conjecture~2]{GKM17} and \cite[Open problem~1]{AFRSV21}). A long line of work sought to prove this conjecture for deterministic algorithms, and more generally understand the worst-case distortion of several well known voting mechanisms. \cite{ABP15} first showed that constant distortion was possible, by showing that the well known Copeland rule guarantees distortion at most 5. This was later improved by \cite{MW19,Kem20b} to $2 + \sqrt{5} \approx 4.236$. Both works also outlined sufficient conditions that would allow their techniques to attain distortion $3$. Finally, building on this work, \cite{GHS20} gave an mechanism that guarantees distortion $3$, thus resolving the conjecture for deterministic mechanisms. A more thorough discussion of this line of work, including bounds for specific voting rules can be found in \cite{AFRSV21}.

With the optimal distortion for deterministic mechanisms resolved, we turn our attention to randomized mechanisms, where much less is known. \cite{AP17} first showed that the Randomized Dictatorship rule (which selects a random voter and chooses their top candidate) can guarantee distortion $3 - 2/n$, which beats any deterministic algorithm (albeit by a vanishingly small amount). Later, \cite{Kem20a} gave an algorithm that guarantees distortion $3 - 2/m$, which remains the best known. Notably, this resolves the optimal distortion for $m = 2$, but leaves a gap for all $m \geq 3$.

\subsection{Our contributions and technical overview}

Our main contribution is proving that no randomized election mechanism can guarantee distortion 2 for elections with more than two candidates, thus disproving the conjecture of \cite[Conjecture~2]{GKM17}. Specifically, we show the following theorem. 

\begin{theorem}\label{thm:LBs}
For each $m \geq 3$, there is a positive constant $\gamma_m$ such that no randomized election mechanism on $m$ candidates can achieve distortion less than $2 + \gamma_m$. In particular, $\gamma_3 \geq 0.02613$, $\gamma_4 \geq 0.04957$, up to $\displaystyle\lim_{m \to \infty} \gamma_m \geq 0.11264$.
\end{theorem}

A more detailed list of the distortion lower bounds for small values of $m$ can be found in Table~\ref{tb:distlbs}.  

Throughout the paper, we use linear programming as a lens through which we view the metric distortion problem. This framing is introduced in Section~\ref{sec:reframing}. While previous work has used LPs in this setting (such as \cite{Kem20b}), our approach is quite different. In \cite{Kem20b}'s LP, the variables are the distances of the hidden metric, and the LP finds the worst metric for a given election mechanism\footnote{We note that \cite{Kem20b}'s LP is suited to deterministic mechanisms, but it can be adapted to randomized mechanisms as well.} (the one which maximizes the distortion). Given this, one can select the mechanism that has the \emph{best} worst-case distortion.

In our LPs, the variables are the weights that the mechanism gives each candidate (before normalization), and each possible metric that is consistent with the election imposes a constraint, essentially indicating that the distortion on that metric must be small. While the fact that the LP has infinitely many constraints may seem like a downside, the fact that these constraints are much easier to reason about and manipulate makes this LP a valuable tool. As it turns out, we only need to consider a small number of the constraints to establish our lower bounds, and our work on upper bounds centers around showing that all but a handful of constraints in the LP can be ignored.

The heart of the lower bound argument is our definition of a simple class of metrics, called \emph{$(0,1,2,3)$-metrics}, that in some sense capture the hardness of a given election instance. Each such metric is biased towards a particular candidate by making the distances from that candidate to the voters small, but making other distances as large as possible. This makes it difficult for a randomized mechanism to balance all of the constraints imposed by these metrics. We show that if we have an election where $M$ is the matrix whose $(i, j)$th entry is the proportion of voters that prefer $i$ over $j$ and $\vect{{\plu}}$ is the vector whose $i$th entry is the proportion of voters that rank $i$ first, then any randomized election mechanism must have high distortion on at least one of the $(0,1,2,3)$-metrics as long as the column sums of $M^{-1}$ are nonnegative, and the expression $\vect{1}^\top M^{-1}(\vect{1} - \vect{{\plu}})$ is small (specifically, if it is at most $2 - \eps$, we get a distortion lower bound of $2 + \eps'$). The simplicity of these metrics makes them appealing options for proving lower bounds not only for this problem, but for other problems involving hidden metric spaces and ordinal rankings.

It may seem a priori that the lower bounds we obtain are artifacts of the precise construction that we use, and not necessarily tight bounds for the problem.
However, in Section~\ref{sec:3UB} we show that for elections with $3$ candidates, we have an upper bound that matches the lower bound given by Theorem~\ref{thm:LBs}. The approach is to show that the only metric constraints that truly matter are those for the $(0,1,2,3)$-metrics, and then we show that the hardest election for these metrics is precisely the one we constructed for our lower bound. We do this by introducing two intermediate classes of metrics, called \emph{biased metrics} and \emph{generalized $(0,1,2,3)$-metrics}. We show that for any election, all constraints besides the biased metric constraints are redundant, and when $m = 3$, we can further reduce the biased metric constraints to the generalized $(0,1,2,3)$-metric constraints. Finally, we show that either only the $(0,1,2,3)$-metric constraints are not redundant, or we can give one candidate 0 weight and reduce to an election with fewer candidates. The last step is to show that the worst case election for the $(0,1,2,3)$-metrics is indeed what we constructed for our lower bound. We do this by fudging a given election in a careful way that only increases the worst case distortion, and brings it into the form of our lower bound constructions. We also show a weaker version of this last step for general $m$ in Section~\ref{sec:0123UB}. Specifically, we show that if the underlying metric space for an election instance is guaranteed to be a $(0,1,2,3)$-metric, then there is a mechanism that guarantees distortion at most $2.1397$. In the proof, we use some manipulation of the constraints using linear algebra and inequalities as an alternative to the fudging approach that does not scale well when $m$ increases.

The main problem left open in our work is whether there exist randomized election mechanisms that can guarantee distortion matching our lower bounds, or even distortion $3 - \varepsilon$ for some constant $\varepsilon$. 
We offer two new mechanisms in Section~\ref{sec:dis} that we believe can be leveraged to resolve this problem. These two mechanisms are based on our LP approach and our conjecture that the $(0,1,2,3)$-metrics truly capture the hardness of the problem.

\section{Preliminaries and notation}\label{sec:prelims}

\textbf{Elections.} Throughout the paper, the elections we consider are over a set of $n$ voters $V$, and $m$ candidates $C$. For ease, we will let $C = [m] = \{1, 2,..., m\}$ and we will index candidates by variables $i, j, k, \ell$, and voters by $u$ and $v$. In general, $n$ should be thought of as very large, and we will have expressions that concern the proportion of voters of a certain type, which we allow to take on real numbered values in $(0, 1)$. These numbers will be attainable for real elections with arbitrarily good accuracy if $n$ is allowed to be sufficiently large.

Let $\plu(i)$ denote the set of voters that rank $i$ first. With some abuse of notation, we will also let $\plu(i)$ denote the proportion of voters that rank $i$ first, and we will let $\vect{{\plu}}$ denote the vector whose $i$th entry is $\plu(i)$. Similarly, let $\plu(I)$ denote the set/proportion of voters that rank the candidates of $I$ above all other candidates. We use $\plu(i, j)$ as shorthand for $\plu(\{i, j\})$ to keep notation lean.

If a voter $v$ prefers candidate $i$ over candidate $j$, we write $i \cg_v j$. For a pattern $\cal{P}$, we let $S_{\cal{P}}$ denote the subset of voters whose preferences fit the pattern $\cal{P}$, and $s_{\cal{P}} = |S_{\cal{P}}|/n$ is the proportion of these voters overall. For example,  $S_{i \cg j}$ is the set of voters that prefer $i$ over $j$, $S_{i, j \cg k}$ is the set of voters that prefer $i$ and $j$ over $k$, and $S_{I \cg j}$ is the set of voters that prefer all the candidates in $I$ over $j$. We use this notation flexibly to make the notation as lean as possible, but it should be clear from context what is being expressed. 

We use $M$ to denote the $m\times m$ \emph{comparisons matrix} of an election, given by $M_{i, j} = s_{i \cg j}$ for $i \neq j$ and $M_{i,i} = 0$ for each $i$. As $s_{i \cg j} + s_{j \cg i} = 1$, $M$ has the property that $M + M^\top = J  - I$, where $J$ is the all ones matrix, and $I$ is the identity matrix. Matrices with this property are known as \emph{generalized tournament matrices} in the literature, since they can be represented the adjacency matrix of a weighted tournament graph, where the weight of edge $(i, j)$ is the proportion of wins $i$ has against $j$.

For an election with underlying distance metric $d$, we denote the social cost of a candidate $i$ to be their average distance to the voters. i.e., 
$$\SC(i, d) := \frac1n \sum_{v \in V} d(i, v).$$
We may just write $\SC(i)$ when the relevant distance metric is clear from context. We note that in the literature, the social cost of a candidate is usually the \emph{total} distance rather than the average. Since $n$ is a constant for a given election instance, this makes no difference. We find the average easier to work with, since we consider proportions of voters. \\

\noindent\textbf{Metric spaces.} Per the standard definition, a metric space is a pair $(\mcal{M}, d)$ of a set $\mcal{M}$ and a distance metric $d: \mcal{M} \times \mcal{M} \to \mathbb{R}_{\geq 0}$ with the following three properties:
\begin{enumerate}[label=(\arabic*)]
	\item Identity of indiscernibles: $d(x, y) = 0 \iff x = y$, 

	\item Symmetry: $d(x, y) = d(y, x)$,

	\item Triangle inequality: $d(x, y) \leq d(x, z) + d(z, y)$.

\end{enumerate}

At times, we may operate as if $\mcal{M} = V\cup C$, but strictly speaking each member of $V \cup C$ just ``occupies'' a point in $\mcal{M}$. This means that we allow distinct members of $V \cup C$ to occupy the same point in the metric space (and therefore have distance 0). This effectively means that we can ignore the identity of indiscernibles axiom. 

The key property we need to be careful about in defining our metrics is the triangle inequality. Typically, we will not explicitly define all of the distances between members of $V \cup C$. For example, we may just define the distances $d(i, v)$ for $i \in C, v \in V$, in which case it is easy to see that the triangle inequality can be satisfied as long as we satisfy the 3-hop constraints $d(j, v) \leq d(j, u) + d(i, u) + d(i, v)$ for all $i, j\in C$ and  $u, v \in V$. 

In general, if not all of the distances are explicitly defined, then the ``full'' metric should be thought of as the \emph{graph distance closure} of what is defined. More specifically, we define an edge-weighted graph with vertices $V \cup C$, and edges where distances have been explicitly defined with weights given by those distances. Then the distance between an arbitrary pair of members in $V \cup C$ is their minimum distance in this graph. As long as the explicitly defined distances satisfy the property that no edge has greater weight than some path between its endpoints (a generalization of the triangle inequality), the resulting metric will be valid.

\section{Reframing the problem}\label{sec:reframing}

Suppose we have a fixed election, with candidates labeled $\{1, 2, ..., m\}$. Let $\mcal{D}$ denote the set of distance metrics $d$ that are consistent with the election, i.e., if $i \cg_v j$ then $d(i, v) \leq d(j, v)$. 

Now, suppose we have a voting rule that chooses candidate $i$ with probability $p_i$. Then for this voting rule to guarantee distortion $1 + \alpha$ on this election, it means that for each $d \in \mcal{D}$, the following linear constraint is satisfied:
\begin{equation}\label{eq:constraint}
\frac{\sum_{i = 1}^m p_i \SC(i, d)}{\min_{j} \SC(j, d)} \leq 1 + \alpha \iff \sum_{i = 1}^m p_i (\SC(i, d) - \min_{j} \SC(j, d)) \leq \alpha \min_{j} \SC(j, d).
\end{equation}
The optimal voting rule chooses the variables $p_i$ so that $\alpha$ can be made as small as possible. To determine the best value of $\alpha$, we can consider the following linear program.
\begin{align*}
\text{maximize} \quad &\sum_{i = 1}^m p_i & \label{eq:LP1} \tag{A}\\
\text{subject to}\quad  \sum_{i = 1}^m  (\SC(i, d) - \min_{j} \SC(j, d)) p_i &\leq  \min_{j} \SC(j, d),  &\forall d \in \mcal{D}\\
p_i& \geq 0, &1 \leq i \leq m 
\end{align*}
This linear program effectively scales each of the original $p_i$s by $\frac1\alpha$, so that minimizing $\alpha$ is the same as maximizing $\sum p_i$. Indeed, it is easy to check that if we have a feasible solution to the LP with $\sum p_i = \beta$, then $p_i \leftarrow p_i/\beta$ satisfies all the inequalities of (\ref{eq:constraint}) with $\beta = \frac{1}{\alpha}$. Similarly, if we have any choice of $p_i$s that satisfies the inequalities of (\ref{eq:constraint}) for some $\alpha$, then $p_i \leftarrow p_i/\alpha$ is a feasible solution to the LP with objective function $\frac{1}{\alpha}$. It follows that optimizing the LP and optimizing the distortion are the same, formally stated as the following proposition.

\begin{proposition}\label{prop:LPequiv}
Suppose we have an election, with a corresponding LP given by \textnormal{(\ref{eq:LP1})}. If the optimal objective of the LP is $\beta$ then the optimal voting mechanism attains distortion $1 + \frac1\beta$ on this election. 
\end{proposition}

This reframing of the problem as a linear program is a key foundation of our approach, and we will refer to the LP (\ref{eq:LP1}) (and variants of it) repeatedly in the paper. The perspective of this LP has benefits both for lower and upper bounds. For lower bounds, we show that by just considering a small number of special constraints, we can design elections where the optimal objective function of (\ref{eq:LP1}) must be small. For upper bounds, we can argue that most of the constraints in (\ref{eq:LP1}) are redundant, and use the structure of the non-redundant constraints to reason about how large the optimal objective must be.

\section{Lower bounds}\label{sec:LBs}

The goal of this section is to prove Theorem~\ref{thm:LBs}. We obtain our lower bounds by considering a restricted class of metrics, and constructing election instances where the LP constraints in (\ref{eq:LP1}) corresponding to these metrics force the objective function to be small. 
These metrics are defined as follows.

\begin{definition}[(0,1,2,3)-metrics]
Given an election, its \emph{$(0,1,2,3)$-metrics} $\{d_1, d_2, ..., d_m\}$ are given by 
$$d_i(i, v) = \begin{cases}
0 & v \in \plu(i)\\
1 & v \notin \plu(i)
\end{cases}$$
and for $j \neq i$, 
$$d_i(j, v) = d_i(i, v) + \begin{cases}
0 & j \cg_v i \\
2 & j \cl_v i
\end{cases}.$$
\end{definition}

The $(0,1,2,3)$-metrics are named as such because all of the distances are either $0, 1, 2,$ or $3$. Figure~\ref{fig:0123metrics} (in Appendix~\ref{sec:appendix}) gives a visualization of the $(0,1,2,3)$-metrics in our lower bound instance for $m = 3$. These metrics may seem contrived at first, but a natural way to arrive at them is as follows. 

We would like to design a metric $d_i$ for each $i$ where $\SC(i, d_i)$ is small, but $\SC(j, d_i) - \SC(i, d_i)$ is as large as possible for each $j \neq i$. These are metrics that give us particularly restrictive constraints in (\ref{eq:LP1}). A natural place to start is to set $d_i(i, v) = 1$ for each $v$, and then set the other distances to be as large as possible subject to the constraints that $j \cg_v i \implies d_i(j, v) \leq d_i(i, v)$ and the triangle inequality which gives us 3-hop constraints $d(j, v) \leq d(j, u) + d(i, u) + d(i, v)$ for any candidate $j$ and voters $u$ and $v$. The first constraints imply that $d_i(j, v) = 1$ whenever $j \cg_v i$, and the second imply that $d_i(j, v) \leq 3$ (as long as there exists some voter $u$ such that $j \cg_u i$, which we can assume because otherwise we can disregard candidate $j$ completely). Setting $d_i(j, v) = 3$ whenever $i \cg_v j$ indeed satisfies all the necessary conditions. We can call these metrics \emph{$(1, 3)$-metrics}. Once we have these metrics defined, it is not hard to see that if for $v \in \plu(i)$ we change $d_i(i, v)$ from 0 to 1 and change $d_i(j, v)$ from 3 to 2, then all of the same conditions are still satisfied. But now, $\SC(i, d_i)$ is smaller, and $\SC(j, d_i) - \SC(i, d_i)$ is unchanged. Thus, these metrics, which are exactly our $(0,1,2,3)$-metrics, give us even tighter constraints for (\ref{eq:LP1}). We note that this adjustment is necessary, since a corollary of Lemma~\ref{lem:clevermatrix} tells us that we can guarantee distortion 2 on $(1,3)$-metrics.

We also note that $(0,1,2,3)$-metrics can be viewed as a generalization of the metrics considered in Figure~\ref{fig:two-metrics}, which are in fact $(0,1,2,3)$-metrics themselves.

It is not difficult to check via casework that the $(0,1,2,3)$-metrics do indeed satisfy the triangle inequality, but we will postpone a formal proof to Section~\ref{ssec:biasdefs}, where we prove this for a larger class of metrics that includes the $(0,1,2,3)$-metrics. 

Next, we consider the LP constraints that are imposed by the $(0,1,2,3)$-metrics. Using the notation introduced in Section~\ref{sec:prelims}, we have that $\min_j \SC(j, d_i) = \SC(i, d_i) = 1 - \plu(i)$, and $\SC(j, d_i) - \SC(i, d_i) = 2s_{i \cg j}$. Thus, the constraint we get for $d_i$ is 
$$2 \sum_{j \neq i} s_{i \cg j}p_j \leq  1 - \plu(i).$$


Since the factor of $2$ can be absorbed into the variables $p_i$, we have the following lemma.

\begin{lemma}\label{lem:0123lb}
Suppose we have an election such that if, 
\begin{equation}\label{eq:0123csts}
\sum_{j \neq i} s_{i \cg j}p_j \leq  1 - \plu(i) 
\end{equation}
for each $1 \leq i \leq m$, then $\sum_i p_i \leq \beta$. Then no randomized election mechanism can guarantee distortion less than $1 + \frac{2}{\beta}$.  
\end{lemma}

\begin{proof}
If the condition of the lemma is satisfied, then it follows that if the $p_i$s satisfy the inequalities $2\sum_{j \neq i} s_{i \cg j}p_j \leq  1 - \plu(i)$ for each $i$, then $\sum_i p_i \leq \beta/2$. Since these inequalities are all constraints in (\ref{eq:LP1}), this implies that the optimal objective of (\ref{eq:LP1}) is at most $\beta/2$. This implies by Proposition~\ref{prop:LPequiv} that the optimal voting rule for this election has distortion at least $1 + 2/\beta$, as desired.
\end{proof}

The remainder of this section is devoted to constructing elections that satisfy the condition of Lemma~\ref{lem:0123lb} with small values of $\beta$.

The inequalities of (\ref{eq:0123csts}) can be written as
$$M\vect{p}  \leq \vect{1} - \vect{{\plu}}.$$
Here, $\vect{p}$ denotes the column vector whose $i$th entry is $p_i$, and $\vect{1}$ denotes the all ones vector. Now, suppose that we construct our election so that $M$ is invertible, and the column sums of $M^{-1}$ (i.e., the entries of the vector $\vect{1}^\top M^{-1}$) are nonnegative. Then it follows that 
$$\vect{1}^\top M^{-1}M\vect{p} \leq  \vect{1}^\top M^{-1}(\vect{1} - \vect{{\plu}}) \implies \vect{1}^\top\vect{p} \leq \vect{1}^\top M^{-1}(\vect{1} - \vect{{\plu}}).$$
Here, the left inequality can be interpreted as multiplying the inequalities in (\ref{eq:0123csts}) by the entries of $\vect{1}^\top M^{-1}$ (which are nonnegative, so the inequality does not flip), and adding up the inequalities. As $\vect{1}^\top\vect{p}$ is just $\sum_i p_i$ this tells us that $\sum_i p_i \leq \vect{1}^\top M^{-1}(\vect{1} - \vect{{\plu}}),$ so we can take this to be $\beta$ in Lemma~\ref{lem:0123lb}. Now we would like to construct elections so that $\vect{1}^\top M^{-1}(\vect{1} - \vect{{\plu}})$ is as small as possible. 

Our elections will have $M$ and $\vect{{\plu}}$ in the following form for some parameters $a, b, c$.

\begin{equation}\label{eq:MandPlu}
M = \begin{bmatrix}
0 & a & a & a & \cdots & a\\
1 - a & 0 & b & b & \cdots & b\\
1 - a & 1 - b & 0 & c & \cdots & c\\
1 - a & 1 -b &  1 - c &  0 & \cdots & \frac12\\
\vdots & \vdots & \vdots & \vdots & \ddots & \vdots \\
1 - a & 1 -b &  1 - c &  \frac12 & \cdots & 0\\
\end{bmatrix} \qquad \qquad 
\vect{{\plu}} = \begin{bmatrix}
a\\
b\\
1 - (a + b)\\
0\\
\vdots \\
0
\end{bmatrix}.
\end{equation}

Here, the bottom right $(m - 3)\times(m - 3)$ submatrix of $M$ is $0$ along the diagonal and $\frac12$ elsewhere. Clearly for this to be attainable by a real election, we must have $a + b + c \geq 1$, since $s_{3 \cg 4} \geq \plu(3)$ implies $c \geq 1 - (a + b)$.

We will argue that when $a + b + c \geq 1$, we can construct elections such that $M$ and $\vect{{\plu}}$ are in the above form.

For $m = 3$, $c$ does not come into play, and the above form can be attained just by having the ranking permutations $(1, 3, 2), (2, 3, 1), (3, 2, 1)$ in proportions $a, b, 1 - (a + b)$ respectively. For $m = 4$, we can get the above form by inserting 4 next to 3 in some careful proportions:

\begin{center}
\begin{tabular}{ c|c }
\text{Permutation} & \text{Proportion}\\   
\hline
(1, 4, 3, 2) & $a\cdot\frac{1 - c}{a + b}$\\
(1, 3, 4, 2) & $a\left(1 - \frac{1 - c}{a + b}\right)$\\
(2, 4, 3, 1) & $b\cdot\frac{1 - c}{a + b}$\\
(2, 3, 4, 1) & $b\left(1 - \frac{1 - c}{a + b}\right)$\\
(3, 4, 2, 1) & $1 - (a + b)$\\
\end{tabular}
\end{center}

Since $a + b + c \geq 1$, all these proportions are nonnegative. Finally, for $m > 4$ we can replace $4$ in the above example with each permutation of $(4, 5, ..., m)$ in equal proportion to get the desired forms of $M$ and $\vect{{\plu}}$.

For ease, let $k = m - 3$. We can check by matrix multiplication that $M^{-1}$ is given by 
$$ \mkern-88mu \frac{1}{(k + 1) - 2kc(1 - c)}\begin{bmatrix}
-(k + 1)\frac{b(1 - b)}{a(1 - a)} & \frac{(k + 1)(1 - b) - 2kc(1 - c)}{1 - a} & \frac{b((k + 1) - 2kc)}{1 - a} & \frac{2bc}{1 - a} & \cdots & \frac{2bc}{1 - a}\\
\frac{(k + 1)b - 2kc(1 - c)}{a} & -(k + 1) & (k + 1) - 2kc & 2c & \cdots & 2c\\
\frac{(1 - b)(2kc - (k - 1))}{a} & (k + 1) - 2k(1 - c) & -(k + 1) & 2(1 - c) & \cdots & 2(1 - c)\\
\frac{2(1 - b)(1 - c)}{a} & 2(1 - c) & 2c &  -2k + 4(k - 1)c(1 - c) & \cdots & 2(1 - 2c(1-c))\\
\vdots & \vdots & \vdots & \vdots & \ddots & \vdots \\
\frac{2(1 - b)(1 - c)}{a} & 2(1- c) &  2c &  2(1 - 2c(1-c)) & \cdots & -2k + 4(k - 1)c(1 - c)
\end{bmatrix}.$$

To be clear, the bottom right $k \times k$ matrix has $-2k + 4(k - 1)c(1 - c)$ along the diagonal, and $2(1 - 2c(1-c))$ elsewhere. We note for the interested reader that if we fix $c = 1/2$, then this matrix (and the math that follows) simplifies significantly, and this is enough to get a lower bound of $\frac{1 + 2\sqrt{7}}{3} \approx 2.0972$ for the distortion. 

The column sums of $M^{-1}$ are easy to compute, since each column besides the first must have dot product 0 with the first row of $M$. This dot product is just $a$ multiplied by the sum of the entries of the column besides the first entry, which means that the sum of those entries is 0. Therefore, the sum of the $i$th column for $2 \leq i \leq m$, is just its first entry. Computing the sum of the first column separately by hand, we have that 
$$\mkern-36mu \vect{1}^\top M^{-1} = \frac{1}{(k + 1) - 2kc(1 - c)}\begin{bmatrix}
-(k + 1)\frac{b(1 - b)}{a(1 - a)} + \frac{(k + 1) - 2kc(1 - c)}{a} & \frac{(k + 1)(1 - b) - 2kc(1 - c)}{1 - a} & \frac{b((k + 1) - 2kc)}{1 - a} & \frac{2bc}{1 - a} & \cdots & \frac{2bc}{1 - a}
\end{bmatrix}.$$
The condition that these entries must be nonnegative imposes some additional constraints on $(a, b, c)$. We note that $(k + 1) - 2kc(1 - c) > 0$ since $2c(1 - c)\leq \frac12$. By ensuring that $b, c < \frac12$, we guarantee that $(k + 1)(1 - b) - 2kc(1 - c) > 0$ and $(k + 1) - 2kc > 0$. Finally, by ensuring that $1 - 2c(1 - c) > \frac{b(1 - b)}{(1 - a)}$, we have that the first column sum is positive.

With all the necessary conditions determined, we can move on to determining the distortion lower bound we get, as a function of $a, b, c$. Recall that $\beta = \vect{1}^\top M^{-1}(\vect{1} - \vect{{\plu}})$, which is
\begin{align*}
\beta(a, b, c, k) = \frac{1}{(k + 1) - 2kc(1 - c)}\left(\frac{-(k + 1)b(1 - b) + ((k + 1) - 2kc(1 - c))(1 - a)}{a}\right.\\
 \left. + \frac{(k + 1)(1 - b)^2 - 2kc(1 - c)(1 - b) + b((k + 1) - 2kc)(a + b)}{1 - a}\right)
\end{align*}
For each fixed value of $k = m - 3$, we can optimize the above function in the region where $a, b, c\in (0, 1)$, $a + b + c \geq 1$, $b, c < \frac12$, and $1 - 2c(1 - c) > \frac{b(1 - b)}{(1 - a)}$ to find the smallest possible $\beta$. This gives us a lower bound of $1 + 2/\beta$ for the distortion. Table~\ref{tb:distlbs} provides a list of these values for several choices of $m$. 

It is not difficult to check that the choices of $a,b,c$ in the table satisfy the required constraints. We note that the condition that $a + b + c \geq 1$ does not apply when $m = 3$. We can easily check the condition $1 - 2c(1 - c) > \frac{b(1 - b)}{(1 - a)}$ by noting that $2c(1 - c) + \frac{b(1 - b)}{(1 - a)}$ is a function that increases as $a, b, c$ increase (since $b, c < \frac12$). Since we always end up choosing $a \leq 0.48, b \leq 0.43$, and $c \leq 0.4$, we have $2c(1 - c) + \frac{b(1 - b)}{(1 - a)} \leq 0.96 < 1$ as needed.

\begin{table}
\centering
\begin{tabular}{ |c|c|c|c|c|c| }
$m$ & $a$ & $b$ & $c$ & $\beta$ & Distortion Lower Bound\\
\hline
3  & 0.473356  & 0.423961  & -  &  1.94907 & 2.02613
\\
4  & 0.459994  & 0.406749  & 0.363254  & 1.90554  & 2.04957
\\
5  & 0.452953  & 0.400474  & 0.373050  & 1.88106  & 2.06323
\\
6  & 0.448571  & 0.397287  & 0.378427  & 1.86556  &2.07206
 \\
7  & 0.445576  & 0.395377  & 0.381835  & 1.85490  &2.07822
 \\
8  & 0.443396  & 0.394112  & 0.38419  & 1.84712  & 2.08276
\\
9  & 0.441738  & 0.393215  & 0.385916  &  1.84120 & 2.08625
\\
10 & 0.440434  & 0.392546  & 0.387236  &  1.83654 & 2.08900
\\
\hline
50 & 0.431584  & 0.388789  & 0.395418  &  1.80493 & 2.10807
\\
100 & 0.430538  & 0.388426  & 0.396305  &  1.80121 & 2.11036
\\
1000 & 0.429608  & 0.388115  & 0.397081  &  1.79790 & 2.11241
\\
$\infty$ & 0.429505  & 0.388082  & 0.397166  &  1.79753 & 2.11264
\end{tabular}
\caption{\label{tb:distlbs} Distortion lower bounds for several values of $m$, including what choices of $a, b, c$ in the admissible region minimize $\beta$.}
\end{table}

Finally, as $m \to\infty$, the function $\beta(a, b, c, k)$ simplifies somewhat to the following.

$$\mkern-24mu\lim_{k \to \infty} \beta(a, b, c, k) =  \frac{1}{1 - 2c(1 - c)}\left(\frac{-b(1 - b) + (1 - 2c(1 - c))(1 - a)}{a}\right.\\
 \left. + \frac{(1 - b)^2 - 2c(1 - c)(1 - b) + b(1 - 2c)(a + b)}{1 - a}\right).$$
Like with $\beta(a,b,c,k)$ for fixed values of $m$, for the above function, it is not hard to determine the values of $a,b,c$ that minimize the function. These are included as well in Table~\ref{tb:distlbs}. Thus, we have established distortion lower bounds all $m \geq 3$, as claimed in Theorem~\ref{thm:LBs}. 

We note that it may seem natural to change more of the rows of $M$ with $\frac12$s to free variables, but numerical analysis suggests that this does not give any improvement. For the interested reader, in Appendix~\ref{sec:appendix}, we go through the case where $m = 3$. There, for variety, we work directly with the definition of distortion, rather than using the LP approach from Section~\ref{sec:reframing}.

\section{Matching upper bound for 3 candidates}\label{sec:3UB}

We will prove the following theorem.

\begin{theorem}
Any election with $3$ candidates has a randomized election mechanism that guarantees distortion at most $2.02613$.
\end{theorem}

The strategy is to argue that our example for $m = 3$ is the worst case. Since our examples are very particular, we take several steps to gradually restrict the elections (and metrics) that we need to consider, until we are left with our lower bound instance. 

\begin{enumerate}[label={(\arabic*)}]

	\item First, we show that all of the constraints in (\ref{eq:LP1}) are redundant, besides those that correspond to what we call \emph{biased metrics} (defined in Section~\ref{ssec:biasdefs}). Moreover, we show that when $m = 3$, these metrics reduce to an even more restricted class, which we call \emph{generalized $(0,1,2,3)$-metrics}.

	\item Next, we consider the LP given by (\ref{eq:LP1}) with only the generalized $(0,1,2,3)$-metrics. We argue that either the $(0,1,2,3)$-metrics are the only non-redundant constraints, or we can automatically get distortion at most 2.

	\item Finally, if the mechanism only has to be concerned about $(0,1,2,3)$-metrics, the distortion that it can achieve only depends on the comparisons matrix and plurality vector. We argue that in the worst case elections, these take the form given in (\ref{eq:MandPlu}). Since in the last section, we optimized over the worst examples of this form, this gives us the theorem.

\end{enumerate}

One could imagine a similar strategy working to show an upper bound for general $m$ that matches our lower bound. The two main sticking points are that the reduction from biased metrics to generalized $(0,1,2,3)$-metrics in step (1) does not work as easily for larger values of $m$, and the analysis of the worst case elections for $(0,1,2,3)$-metrics in step (3) involves a lot of casework and consideration of the extrema of multivariate functions, and this does not scale well as $m$ increases. On the bright side, step (2) does essentially work for general $m$ (rather than distortion at most $2$, it is distortion at most $2 + \gamma_{m - 1}$, per the notation in Theorem~\ref{thm:LBs}). We also show in Theorem~\ref{thm:0123UB} that the casework in step (3) can be averted without giving up much, by showing that if a mechanism only needs to be concerned about $(0,1,2,3)$-metrics, it can guarantee distortion at most $2.13968$ (leaving only a small gap with our asymptotic lower bounds). 

We will aim to prove each step in as much generality as is allowed by our techniques. While it may make some steps more involved, hopefully it is also more informative.

\subsection{Biased metrics and generalized $(0,1,2,3)$-metrics}\label{ssec:biasdefs}

We will define a set of metric spaces whose constraints are sufficient to consider in the LP (\ref{eq:LP1}). The rough idea is to show that if we have an arbitrary metric, then we can massage it into one of our \emph{biased} metrics so that $\min_j \SC(j)$ decreases, and $\SC(i) - \min_j \SC(j)$ increases. This means that the constraint imposed by the new metric is tighter than that imposed by the old metric, which means that the latter constraint is redundant. 

\begin{definition}
Let $(x_1, ..., x_m)$ be a vector of nonnegative real numbers such that at least one entry is 0. Given an election, the \emph{biased} metric for the vector $(x_1, ..., x_m)$ is defined as follows. For a voter $v$, define
$$y_v = \frac12 \max_{i, j: i \cgeq_v j} (x_i - x_j).$$
Then for each candidate $i$ and voter $v$, we have 
$$d(i, v) = y_v + \min_{j: i \cgeq_v j} x_j.$$
\end{definition}

We note that the $(0,1,2,3)$-metric $d_i$ corresponds to the vector that is $0$ at $i$ and $2$ elsewhere. The intuition behind these metrics is the following. Suppose we had an arbitrary metric so that $i^* = \arg\min_i\SC(i)$. If we were to throw out all of the distances besides the distances $d(i^*, i) =: x_i$ between $i^*$ and other candidates, and then (1) made the distances $d(i^*, v) =: y_v$ as small as possible without breaking any triangle inequality constraints, and (2) made the distances $d(i, v)$ for $i \neq i^*$ as large as possible without breaking any triangle inequality or ordinal constraints, then the result would be exactly the above metric. We call these metrics ``biased'' because they are designed to take a given metric, and then maximize the bias towards the candidate $i^*$.

It is not immediately clear that biased metrics satisfy the conditions required of a metric space, so we will prove this explicitly below. 

\begin{proposition}
For any vector $(x_1, ..., x_m)$ and any election, the biased metric for $(x_1, ..., x_m)$ is indeed a valid distance metric. 
\end{proposition}

\begin{proof}
Clearly, the distances we have defined are nonnegative, since the expression for $y_v$ allows for $i = j$, which means that $y_v$ is the maximum over a set that includes 0. Thus, it suffices to show that biased metrics satisfy the triangle inequality.

Recall that we view the partially defined metric as a weighted graph, and the complete metric is the graph distance metric on this graph. This means that it suffices to show that for any explicitly defined edge (i.e., between a candidate $i$ and a voter $v$), no path between the two endpoints can have smaller total weight than the defined weight of the edge. 

Suppose we have some path between a candidate $i$ and voter $v$. We will show that $d(i, v)$ is at most the total weight of the path. Note that $d(i, v) = y_v + \displaystyle\min_{j: i \cgeq_v j} x_j \leq y_v + x_i$. Now, suppose that the first two edges on the path are $(i, u)$ and $(u, j)$, and the last edge is $(k, v)$. Then the total length of the path is at least 
$$d(i, u) + d(u, j) + d(k, v) \geq d(i, u) + y_u + y_v.$$
Now, $d(i, u) = y_u + x_\ell$ for some $\ell$ for which $i \cgeq_u \ell$. By the definition of $y_u$, we also have $2y_u \geq x_i - x_\ell$. Putting all of these together, we have
$$d(i, u) + y_u + y_v \geq 2y_u + x_\ell + y_v \geq x_i + y_v \geq d(i, v)$$
as desired.
\end{proof}

Next, we will show that for the purpose of the LP (\ref{eq:LP1}), it suffices to only consider the constraints imposed by biased metrics.

\begin{proposition}\label{prop:suffbiased}
Suppose we have an election and a distance metric $d$ that is consistent with the election. Let $i^* = \arg\min_i\SC(i, d)$. Then there is a biased metric $\hat{d}$ such that $\SC(i^*, \hat{d}) \leq \SC(i^*, d)$ and  $\SC(i, \hat{d}) - \SC(i^*, \hat{d}) \geq \SC(i, d) - \SC(i^*, d)$ for each $i \neq i^*$.
\end{proposition}

\begin{proof}
As suggested earlier, let $x_i = d(i, i^*)$, and let $\hat{d}$ be the biased metric for $(x_1,x_2, ..., x_m)$. We will show that for any voter $v$, $\hat{d}(i^*, v) \leq d(i^*, v)$ and  $\hat{d}(i, v) - \hat{d}(i^*, v) \geq d(i, v) - d(i^*, v)$, which will immediately imply the proposition.

Fix $v$, and let $i$ and $j$ be such that $i \cgeq_v j$. Then we have
$$ d(i, i^*) \leq d(i, v) + d(i^*, v) \leq d(j, v) + d(i^*, v) \leq d(j, i^*) + 2d(i^*, v).$$
This implies that $d(i^*, v) \geq  \frac{d(i, i^*) - d(j, i^*)}{2} = \frac{x_i - x_j}{2}$. Taking the maximum over all choices of $i$ and $j$ such that $i \cgeq_v j$, we get 
$$d(i^*, v) \geq \frac{1}{2}\max_{i, j: i \cgeq_v j} (x_i - x_j) = y_v = \hat{d}(i^*, v)$$ 
as claimed. Next, fix a candidate $i \neq i^*$ and a voter $v$. Let $j = \displaystyle\arg \min_{j: i \cgeq_v j} x_j$, so that $\hat{d}(i, v) - \hat{d}(i^*, v) = x_j$. Then we have
$$d(i, v) \leq d(j, v) \leq d(j, i^*) + d(i^*, v) = x_j + d(i^*, v)$$
which means that 
$$ d(i, v) - d(i^*, v) \leq x_j = \hat{d}(i, v) - \hat{d}(i^*, v)$$
as claimed. This establishes the proposition.
\end{proof}

Next, we define a class of metrics that sits between biased metrics and $(0,1,2,3)$-metrics, called \emph{generalized $(0,1,2,3)$-metrics}. This class has the convenience that we can prove a kind of reduction from generalized $(0,1,2,3)$-metrics to $(0,1,2,3)$-metrics, and in the case where $m = 3$, we can prove a kind of reduction from biased metrics to generalized $(0,1,2,3)$-metrics.

\begin{definition}
Given an election, its \emph{generalized $(0,1,2,3)$-metrics} are the biased metrics for vectors $(x_1, ..., x_m)$ where for some nontrivial subset $I$ of the candidates, $x_i = 0$ if $i \in I$ and $x_j = 2$ for $j \notin I$. More concretely, the metric $d_I$ has distances as follows. If $i \in I$ 
$$d_I(i, v) = \begin{cases}
0 & v \in \plu(I)\\
1 & v \notin \plu(I)
\end{cases}.$$
Note that $d(i, v) = d(i', v)$ for any $i, i' \in I$. Then for each $j \notin I$ 
$$d_I(j, v) = d_I(i, v) + \begin{cases}
0 & v \notin S_{I \cg j}\\
2 & v \in S_{I \cg j}
\end{cases}$$
for any arbitrary $i \in I$.
\end{definition}

The generalized $(0,1,2,3)$-metrics are what we would get if we collapsed all the candidates of $I$ into one candidate who appears in each voter's ranking at the position where the worst candidate of $I$ was, and then we considered the $(0,1,2,3)$-metric in this new election on fewer candidates. The fact that generalized $(0,1,2,3)$-metrics are $(0,1,2,3)$-metrics in disguise makes it believable that we can reduce from the former to the latter. We prove this formally in the next subsection.

Before we switch to considering generalized $(0,1,2,3)$-metrics alone, we show that when $m = 3$, the biased metrics can be reduced to generalized $(0,1,2,3)$-metrics.

\begin{lemma}\label{lem:biasedtogen}
For any 3 candidate election, the constraints in \textnormal{(\ref{eq:LP1})} corresponding to generalized $(0,1,2,3)$-metrics imply the constraints for all biased metrics.
\end{lemma}

\begin{proof}
First, we note that the generalized $(0,1,2,3)$-metric constraint for the set $I$ can be written in the following form:
$$2\sum_{j \notin I} s_{I \cg j} p_j \leq 1 - \plu(I).$$
This follows since $\SC(i) = 1 - \plu(I)$ for all $i \in I$, and, fixing some arbitrary $i \in I$, $d(j, v) - d(i, v)$ is 2 when $v \in S_{I \cg j}$, and 0 otherwise. To be concrete in the case where $m = 3$, these correspond to the following 6 constraints
$$ 2(s_{1 \cg 2}p_2 + s_{1 \cg 3}p_3) \leq 1 - \plu(1) \qquad\qquad  2(s_{2 \cg 1}p_1 + s_{2 \cg 3}p_3) \leq 1 - \plu(2) \qquad\qquad 2(s_{3 \cg 1}p_1 + s_{3 \cg 2})p_2 \leq 1 - \plu(2)$$
$$ 2\plu(2,3)p_1 \leq 1 - \plu(2,3) \qquad\qquad  2\plu(1,3)p_2 \leq 1 - \plu(1,3) \qquad\qquad 2\plu(1,2)p_3 \leq 1 - \plu(1,2).$$
Next, suppose that we have a biased metric $d$ corresponding to the vector $(x_1, x_2, x_3)$, and without loss of generality suppose that $0 = x_1 \leq x_2 \leq x_3$. We will show that the corresponding constraint is a (nonnegative) linear combination of the constraints 
$$ 2(s_{1 \cg 2}p_2 + s_{1 \cg 3}p_3) \leq 1 - \plu(1) \qquad \text{and} \qquad 2\plu(1,2)p_3 \leq 1 - \plu(1,2).$$
Recall that the biased metric constraint is given by 
\begin{equation}\label{eq:biascst}
(\SC(2) - \SC(1))p_2 + (\SC(3) - \SC(1))p_3 \leq \SC(1).
\end{equation}
By the definition of the biased metric, we have that $d(2, v) - d(1, v)$ is $0$ if $v \in S_{2 \cg 1}$, and is $x_2$ otherwise. This means that 
$$\SC(2) - \SC(1) = s_{1 \cg 2}x_2.$$
Similarly, $d(3, v) - d(1, v)$ is $0$ if $v \in S_{3 \cg 1}$, is $x_2$ if $v \in S_{1\cg 3 \cg 2}$, and is $x_3$ if $v \in S_{1,2\cg 3} = \plu(1,2)$. This means that 
$$\SC(3) - \SC(1) = s_{1\cg 3 \cg 2}x_2 + \plu(1,2)x_3 = s_{1 \cg 3}x_2 + \plu(1, 2)(x_3 - x_2).$$
Thus, the left side of (\ref{eq:biascst}) is 
$$x_2( s_{1 \cg 2}p_2 + s_{1 \cg 3}p_3) + (x_3 - x_2)\plu(1,2)p_3.$$
Next, we consider $\SC(1)$. We have that $d(1, v)$ is $ \frac12 x_3$ if $v \in S_{3 \cg 1}$, $\frac12 x_2$ if $v \in S_{2 \cg 1 \cg 3}$, $\frac12(x_3 - x_2)$ if $v \in S_{1 \cg 3 \cg 2}$, and 0 otherwise. It follows that 
\begin{align*}
\SC(1) &= \frac{x_3}{2} s_{3 \cg 1} + \frac{x_2}{2} s_{2 \cg 1 \cg 3} + \frac{x_3 - x_2}{2} s_{1 \cg 3 \cg 2}\\
&= \frac{x_2}{2} (s_{3 \cg 1} + s_{2 \cg 1 \cg 3}) + \frac{x_3 - x_2}{2}(s_{3 \cg 1} + s_{1 \cg 3 \cg 2})\\
&= \frac{x_2}{2} (1 - \plu(1)) + \frac{x_3 - x_2}{2}(1 - \plu(1, 2)).
\end{align*}  
Thus, (\ref{eq:biascst}) can be written as 
$$x_2( s_{1 \cg 2}p_2 + s_{1 \cg 3}p_3) + (x_3 - x_2)\plu(1,2)p_3 \leq \frac{x_2}{2} (1 - \plu(1)) + \frac{x_3 - x_2}{2}(1 - \plu(1, 2))$$
which is indeed a linear combination of the two claimed constraints, with nonnegative coefficients $x_2$ and $x_3 - x_2$.
\end{proof}

\subsection{Reducing from generalized $(0,1,2,3)$-metrics to $(0,1,2,3)$-metrics}\label{ssec:genreduc}
Say that a particular election mechanism \emph{optimizes over} a certain set of metrics, if it works by solving the LP given by (\ref{eq:LP1}), but only with the constraints corresponding to the special set of metrics, and then normalizing the optimal choice of $\vect{p}$. What we have shown so far is that the election mechanism that optimizes over the biased metrics is in fact optimal for every election instance, and the mechanism that optimizes over the generalized $(0,1,2,3)$-metrics is optimal for elections with 3 candidates.

The LP that optimizes over the generalized $(0,1,2,3)$-metrics has constraints that correspond to $(0,1,2,3)$-metrics, and constraints that do not. The constraints that are tight at the optimal solution determine what the optimal objective is. Next, we will show that either the tight constraints are exactly the $(0,1,2,3)$-metric constraints, or the optimal objective is high anyway. 

\begin{lemma}\label{lem:genreduction}
Over all elections with $m$ candidates, let $\beta_m$ be the minimum optimal objective function value attained by the LP that optimizes over generalized $(0,1,2,3)$-metrics. For a fixed election instance, at the optimum of this LP, either all the $(0,1,2,3)$-metric constraints are tight, or the optimum objective is at least $\beta_{m - 1}$.
\end{lemma}

This lemma effectively means that if we have reduced the problem to considering generalized $(0,1,2,3)$-metrics, then we can actually just consider $(0,1,2,3)$-metrics instead. Another interpretation is that if we construct a lower bound example using generalized $(0,1,2,3)$-metrics, then we can construct an equally good lower bound example with just $(0,1,2,3)$-metrics (by applying this lemma iteratively, and eventually finding the tight constraints that correspond to $(0,1,2,3)$-metrics). 

In the special case where $m = 3$, the lemma tells us that for a given election instance, either the $(0,1,2,3)$-metric constraints are tight, or we can guarantee distortion 2, since we know that we can guarantee distortion 2 in the two candidate case.

\begin{proof}
Recall that the LP that optimizes over the generalized $(0,1,2,3)$-metrics is the following.
\begin{align*}
\text{maximize} \quad &\sum_{i = 1}^m p_i & \\ 
\qquad\qquad\qquad\qquad\text{subject to}\quad  \sum_{j\notin I}  s_{I \cg j} p_j &\leq 1 - \plu(I),  &\forall I\subseteq [m]; I\neq \varnothing, [m]\\
p_i& \geq 0, &1 \leq i \leq m 
\end{align*}
As a reminder the $(0,1,2,3)$-constraints are those that correspond to singleton sets $I = \{i\}$. 

If, for a given election instance, not all of the $(0,1,2,3)$-constraints are tight, then this means that there is at least one such constraint that we can remove without changing the optimum objective of the LP. Suppose one such constraint corresponds to $I = \{i\}$. Now, consider the election obtained by erasing $i$ from each voter's ranking list, and the corresponding LP that optimizes over the generalized $(0,1,2,3)$-metrics. By the definition of $\beta_{m-1}$, this LP has a feasible solution $\vect{p}^*$ that attains objective at least $\beta_{m - 1}$. We claim that the solution that sets $p_i = 0$ and $p_j = p_j^*$ for $j \neq i$ is a feasible solution to the original LP. Since this solution has objective at least $\beta_{m - 1}$, this would complete the proof.

We compare the constraint in the original LP corresponding to $I$ to the constraint in the new LP corresponding to $I \setminus i$ (unchanged if $i \notin I$). Let $\hat{\plu}(I \setminus i)$ denote the fraction of voters that rank the candidates in $I \setminus i$ above all others in the election where $i$ is erased. Then clearly
$$s_{I \cg j} \leq s_{I\setminus i \cg j} \qquad \text{and} \qquad \plu(I) \leq \hat{\plu}(I\setminus i)$$
since if a voter $v$ prefers all the candidates of $I$ over $j$, then they prefer the candidates of $I\setminus i$ over j, and if a voter ranks $I$ above all other candidates, then they rank $I\setminus i$ above all other candidates when $i$ is erased. Since we fix $p_i = 0$, this means that the new constraint is stricter than the old constraint, which means that the claimed solution is indeed a feasible solution for the old LP.
\end{proof}

\subsection{Worst elections for $(0,1,2,3)$-metrics}\label{ssec:worst0123}

Lemma~\ref{lem:genreduction} tells us that it suffices to argue that when the $(0,1,2,3)$-metric constraints are tight, the objective function is sufficiently large. Recall that the $(0,1,2,3)$-metric constraints correspond to the inequality $M\vect{p} \leq \vect{1} - \vect{{\plu}}$, so for these constraints to be tight, we have equality. Now, since we are working in the case when $m = 3$, $M$ can be written in the following form: 
$$M = \begin{bmatrix}
0 & x & 1 - y\\
1 - x & 0 & z\\
y & 1 - z & 0
\end{bmatrix}.$$
Thus, $\det(M) = xyz + (1 - x)(1 - y)(1 - z)$, which, since $x,y,z \in [0, 1]$ means that $M$ is singular only if one of $x,y,z$ is 0, and another is 1. We can assume that this is not the case, since this would mean that $s_{i \cg j} = 1$ for some $i, j \in \{1,2,3\}$, and so we can set $p_j = 0$, remove $j$, and reduce to the two candidate case. This means that we can assume that $M$ is invertible.

It follows that $\vect{p} = M^{-1}( \vect{1} - \vect{{\plu}})$, and so $\sum_i p_i = \vect{1}^\top M^{-1}( \vect{1} - \vect{{\plu}})$. Our goal is to identify the election instance that makes this as small as possible. 

We note that the expression $\vect{1}^\top M^{-1}( \vect{1} - \vect{{\plu}})$ only depends on the comparisons matrix and plurality vector of an election. The idea is to massage a given election instance in a way that only decreases the value of the expression, until $M$ and $\vect{{\plu}}$ are in the form of (\ref{eq:MandPlu}). Since we optimized examples of this form to obtain our lower bound, we will be done.

Suppose we have a fixed election instance. Let $w_\sigma$ denote the proportion of voters whose ranking list is $\sigma$, and let $\sigma'$ denote the reverse permutation of $\sigma$. Suppose that for two permutations $\sigma_1$ and $\sigma_2$, we modify the election instance by increasing $w_{\sigma_1}$ and $w_{\sigma_1'}$ by some $\Delta > 0$ and decreasing $w_{\sigma_2}$ and $w_{\sigma_2'}$ by $\Delta$. Then this does not change $M$, and it increases $\vect{{\plu}}$ for the first and last elements of $\sigma_1$, and decreases it for the first and last elements of $\sigma_2$ (two of these cancel out when $m = 3$). It is possible that this could increase $\vect{1}^\top M^{-1}( \vect{1} - \vect{{\plu}})$, but if instead we decrease $w_{\sigma_1}$ and $w_{\sigma_1'}$ and increase $w_{\sigma_2}$ and $w_{\sigma_2'}$, exactly the opposite effect would happen. Thus, it follows that if there are permutations $\sigma_1, \sigma_1'$ and $\sigma_2, \sigma_2'$ that are distinct and have positive support in an election, then there is a way to modify the election so that $\vect{1}^\top M^{-1}( \vect{1} - \vect{{\plu}})$ decreases, and the support of one of these four becomes 0. Henceforth, we can just consider elections where no such quadruple of permutations all have positive support.

If we consider the different subsets of permutations that may have positive support, there are 12 possibilities. However, without loss of generality we can consider two possibilities the same if they are equivalent up to permutation of the candidates. With this consideration, it suffices to consider only the following three possibilities. The permutations with positive support are highlighted in red.


$${123} \longleftrightarrow \red{321}  \qquad\qquad {123} \longleftrightarrow \red{321} \qquad\qquad \red{123} \longleftrightarrow {321}$$
$$\red{132} \longleftrightarrow \red{231} \qquad \qquad \red{132} \longleftrightarrow \red{231} \qquad \qquad \red{132} \longleftrightarrow \red{231} $$
$${213} \longleftrightarrow \red{312} \qquad\qquad \red{213} \longleftrightarrow {312} \qquad\qquad \red{213} \longleftrightarrow 312$$

We will split into three different cases, one for each possibility above. The three cases will all work in essentially the same way. By considering which permutations have nonzero support, we can determine the entries of $\vect{{\plu}}$ in terms of the entries of $M$. Since $M$ only has 3
variables, we will be able to express $\vect{1}^\top M^{-1}( \vect{1} - \vect{{\plu}})$ as a  function of 3 variables. This is sufficiently simple that we can explicitly optimize the expression using standard techniques. Strictly speaking, this solves the problem without directly having to reduce to the form of (\ref{eq:MandPlu}). Nevertheless, in all three cases, the local optimum ends up being suboptimal, and the expression is minimized at the boundary conditions. These boundary conditions eliminate a variable, and in all but one case, the result corresponds exactly to the form in (\ref{eq:MandPlu}) (possibly with some relabeling of the candidates). Thus, we do reduce to the form of (\ref{eq:MandPlu}), if indirectly. \\

\textbf{Case I.} $w_{123} = w_{213} = 0$.

In this case, the nonzero elements are $w_{321}, w_{132}, w_{231}, w_{312}$. We then have that $w_{132} = \plu(1) = s_{1 \cg 3} = 1 - y$ and $w_{231} = \plu(2) = s_{2 \cg 3} = z$. Therefore,
$$\vect{{\plu}} = \begin{bmatrix}
1 - y \\
z \\
y - z
\end{bmatrix}.$$
With this we can compute that,  
$$ \vect{1}^\top M^{-1}( \vect{1} - \vect{{\plu}}) = 
\frac{x^2 (-y+z+1)-x (y-z-1) (y+z-2)+y (2 y-3)+(z-1) z+2}{xyz + (1 - x)(1 - y)(1 - z)}$$
which we would like to minimize when $x, y, z \in [0, 1]$, $x + y \geq 1$, $x + z \leq 1$, and $y \geq z$. These conditions come from the fact that $\plu(i)$ is at least all the off diagonal entries of $M$ on the $i$th row, and $\plu(i)$ is nonnegative. Taking the partial derivatives and setting them to zero, we find one local minimum at $x = 1/2$, $y = 3 - \sqrt{6}$, and $z = -2 + \sqrt{6}$. The function at this point is $-\frac12 + \sqrt{6} \approx 1.9495$, which is just larger than what we can get at the boundary, which is approximately $1.9491$ (see the first row of Table~\ref{tb:distlbs}). \\


\textbf{Case II.} $w_{123} = w_{312} = 0$.

In this case, the nonzero elements are $w_{321}, w_{132}, w_{231}, w_{213}$. We then have that $w_{132} = \plu(1) = s_{1 \cg 2} = x$ and $w_{231} + w_{213} = \plu(2) = s_{2 \cg 3} = z$. Therefore,
$$\vect{{\plu}} = \begin{bmatrix}
x \\
z \\
1 - (x + z)
\end{bmatrix}.$$
With this we can compute that,  
$$ \vect{1}^\top M^{-1}( \vect{1} - \vect{{\plu}}) = 
\frac{x^2 (y+z-1)+x^3+x (-y z+z-1)-(y (z-1)-2 z) (y+z)-2 y-3 z+2}{xyz + (1 - x)(1 - y)(1 - z)}$$
which we would like to minimize when $x, y, z \in [0, 1]$, $x + y \leq 1$, $x + z \leq 1$, and $x + y + z \geq 1$. Similar to Case I, taking partial derivatives we find a local minimum at $x=y=z=\frac13$, where the expression above is 2. Thus, the global minimum is again at the boundary conditions.\\

\textbf{Case III.} $w_{321} = w_{312} = 0$.

In this case, the nonzero elements are $w_{123}, w_{132}, w_{231}, w_{213}$. We then have that $w_{123} + w_{132} = \plu(1) = s_{1 \cg 2} = x$ and $w_{231} + w_{213} = \plu(2) = s_{2 \cg 1} = 1 - x$. Therefore,
$$\vect{{\plu}} = \begin{bmatrix}
x \\
1 - x \\
0
\end{bmatrix}.$$
With this we can compute that,  
$$ \vect{1}^\top M^{-1}( \vect{1} - \vect{{\plu}}) = 
\frac{x^2 (y-z+2)+x ((y-1) y-(z-3) z-3)+y (z-1)+(z-2) z+2}{xyz + (1 - x)(1 - y)(1 - z)}$$
which we would like to minimize when $x, y, z \in [0, 1]$, $x + y \leq 1$ and $x + z \geq 1$. Unlike the other cases, there is no local minimum in this region, and fixing the boundary conditions does not quite reduce to the desired form. However, we can find that even after fixing $y = 1 - x$ or $z = 1 - x$, there is still no local minimum in the region, so we must set $y = z = 1 - x$. At this point, the function reduces to just $\frac{1}{x(1 - x)} - 2$, which attains its minimum of 2 at $x = \frac12$. Thus, the minimum of the function in the region is 2 at $x=y=z=\frac12$.\\

Thus, we can finally conclude that for $3$ candidate elections, our lower bound instances are indeed the worst case, and that for any such election, there is a randomized election mechanism that guarantees distortion at most $2.02613$.

\section{Nearly matching upper bound for $(0,1,2,3)$-metrics}\label{sec:0123UB}

In this section, we will prove the following theorem. 

\begin{theorem}\label{thm:0123UB}
The LP that optimizes over the $(0,1,2,3)$-metrics always has optimal objective at least $1.75487$. Equivalently, under the assumption that the underlying metric is a $(0,1,2,3)$-metric, there is an election mechanism that guarantees distortion $2.13968$.
\end{theorem}

Here are two helpful ways to interpret this theorem. First, it shows that the election instances of the form in (\ref{eq:MandPlu}) are just about the best possible for constructing lower bound instances with $(0,1,2,3)$-metrics. Second, it shows that if one can prove reductions similar to those done in Sections~\ref{ssec:biasdefs} and~\ref{ssec:genreduc} for cases beyond $m = 3$, then we can recover distortion upper bounds that are close to optimal, without going into the types of casework heavy arguments in Section~\ref{ssec:worst0123}.

Recall that the LP that optimizes over the $(0,1,2,3)$-metrics is of the following form.
\begin{align*}
\qquad\qquad\qquad\qquad\text{maximize} \quad &\sum_{i = 1}^m p_i & \label{eq:LP0123} \tag{B}\\
\text{subject to}\quad  M\vect{p} &\leq \ones - \vect{{\plu}},  \\
\vect{p}& \geq 0.
\end{align*}
We note that like with our lower bounds, we drop the factor of 2 that would appear if we strictly followed (\ref{eq:LP1}), and so an optimal objective of $\beta$ translates to distortion $1 + 2/\beta$.

By the same argument as in Lemma~\ref{lem:genreduction}, it suffices to argue that when the main constraints ($M\vect{p} \leq \ones - \vect{{\plu}}$) are tight, the optimal objective is sufficiently large (otherwise, we can set $p_i = 0$ if the constraint for the metric $d_i$ is not tight, and reduce to a case where $m$ is smaller). In particular, we want to argue that if $\vect{p}$ is a vector such that $M\vect{p} = \ones - \vect{{\plu}}$, then $\sum_i p_i$ is large. 

We start with the following lemma. All norms are the $L^2$ norm. 

\begin{lemma}\label{lem:clevermatrix}
Let $M$ be a matrix with nonnegative values such that $M + M^\top = J - I$, and let $\vect{x}$ be a nonnegative vector with $\|\vect{x}\| \leq 1$. Let $\vect{p}$ be a nonnegative vector such that $M\vect{p} = 1 - \vect{x}$. Then 
$$\sum_{i} p_i \geq 1 + \sqrt{1 - \|\vect{x}\|^2}.$$
\end{lemma}

The expected use case here is when $\vect{x}$ is the plurality vector. This more general version has the interesting corollary that when $\vect{x} = \vect{0}$, it tells us that under the assumption that the underlying metric is a $(1,3)$-metric (which we defined as a precursor to $(0,1,2,3)$-metrics), then distortion $2$ is possible. This means that the slight modification we made to go from $(1,3)$-metrics to $(0,1,2,3)$-metrics was necessary, and that this is the crux of why we are able to squeeze out better lower bounds.

\begin{proof}
Let $s = \sum_{i} p_i.$ We have 
\begin{align*}
M + M^\top &= J - I \\
\implies \vect{p}^\top(M + M^\top)\vect{p} &= \vect{p}^\top J \vect{p} - \vect{p}^\top \vect{p} \\
\implies 2\vect{p}^\top(1 - \vect{x}) &= s^2 - \|\vect{p}\|^2\\
\implies \|\vect{p}\|^2 - 2\vect{p}^\top \vect{x} &=  s^2 - 2s .
\end{align*}
By the Cauchy-Schwarz inequality, we have that $\vect{p}^\top \vect{x} \leq \|\vect{p}\| \cdot \|\vect{x}\|$. This means that $s^2 - 2s \geq \|\vect{p}\|^2 - 2\|\vect{p}\| \cdot \|\vect{x}\|$. On the other hand, $0 \leq \|\vect{p}\| \leq s$, so we can write $\|\vect{p}\| = \lambda s$ for some $\lambda \in [0, 1]$. From this we get that 
$$s^2 - 2s \geq \lambda^2 s^2 - 2\lambda s \|\vect{x}\| \implies s \geq 2\frac{1 - \lambda \|\vect{x}\|}{1 - \lambda^2}.$$
Now, we can think of $\|\vect{x}\|$ as some constant $c$. Given this, we can find the minimum value of the function $f(\lambda) = \frac{1 - c\lambda}{1 - \lambda^2}$ to lower bound the expression in the expression above. 

It is easy to check that $f$ is convex on $[0, 1]$ since $c \leq 1$, and so it is minimized where the derivative is zero. Since $f'(\lambda) = \frac{2\lambda - c(\lambda^2 + 1)}{(1 - \lambda^2)^2}$, this happens when $2\lambda = c(\lambda^2 + 1)$. Solving for $\lambda$, we get $\lambda = \frac{1 - \sqrt{1 - c^2}}{c}$. Substituting back into $f$, we get 
$$\frac{1 - c\lambda}{1 - \lambda^2} =  \frac{1 - (1 -  \sqrt{1 - c^2})}{1 - \left(\frac{1 - \sqrt{1 - c^2}}{c}\right)^2}  = \frac{c^2\sqrt{1 - c^2}}{c^2 - (1  -\sqrt{1 - c^2})^2} = \frac{c^2\sqrt{1 - c^2}}{2(\sqrt{1 - c^2} - (1 - c^2))}$$
$$= \frac{c^2}{2(1 - \sqrt{1 - c^2})} = \frac12(1 + \sqrt{1 - c^2}).$$
And so we get 
$$s \geq  1 + \sqrt{1 - c^2}$$
as claimed.
\end{proof}

Next, we note that the solution that sets $p_i = \frac{\plu(i)}{ 1- \plu(i)}$ is feasible for (\ref{eq:LP0123}), since 
$$\sum_{j \neq i} s_{i \cg j}\cdot \frac{\plu(j)}{ 1- \plu(j)} \leq \sum_{j \neq i} \plu(j) = 1 - \plu(i).$$
Here, we use the fact that $s_{i \cg j} \leq 1- \plu(j)$. This solution corresponds to the \textsc{SmartDictatorship}  mechanism studied by \cite{GHS20}. Since $f(x) = \frac{1}{1 - x}$ is convex on $(0,1)$, by Jensen's inequality with weights $\plu(i)$, we get 
\begin{equation}\label{eq:jensenapp}
\sum_{i} \frac{\plu(i)}{ 1- \plu(i)} \geq \frac{1}{1 - \|\vect{{\plu}}\|^2}.
\end{equation}
We have exhibited a feasible solution with objective function at least $\frac{1}{1 - \|\vect{{\plu}}\|^2}$, so in conjunction with Lemma~\ref{lem:clevermatrix} with $\vect{x} = \vect{{\plu}}$, we get that the optimal objective of (\ref{eq:LP0123}) is at least  
$$\max\left(1 + \sqrt{1 - \|\vect{{\plu}}\|^2}, \frac{1}{1 - \|\vect{{\plu}}\|^2} \right).$$
Letting $y = 1 - \|\vect{{\plu}}\|^2$, we see that the function $1 + \sqrt{y}$ is $1$ at $0$ and increases, whereas $\frac1y$ is $\infty$ at 0 and decreases. This means that the maximum of the two functions is minimal at their intersection point, which is at the real root of $y^3 - y^2 + 2y - 1$ at approximately $0.56984$. At this point, we get an optimal objective of at least $1.75487$, which corresponds to distortion $2.13968$, as claimed in Theorem~\ref{thm:0123UB}. \\

We note that a very slightly better bound can be obtained by proving a slightly tighter inequality than (\ref{eq:jensenapp}). Specifically, one can show that 
$$ \sum_{i} \frac{\plu(i)}{ 1- \plu(i)} \geq 3\|\plu\|^2 + \frac12.$$
This can be proven with a careful application of Lagrange multipliers, and shows that the optimal objective is at least $\frac{5 + \sqrt{31}}{6} \approx 1.76129$, and so the optimal distortion is at most $2.13553$. We omit the tedious details since the improvement is marginal.

The remaining slack in the argument appears to be in the application of the Cauchy-Schwarz inequality, since equality occurs when $\vect{p}$ and $\vect{{\plu}}$ are scalar multiples of each other, which is not the case in our lower bound examples. This seems hard to tighten without getting into the messy details of the weights of different ranking permutations like we did in Section~\ref{ssec:worst0123}.

\section{Discussion}\label{sec:dis}

Though we have proven new lower bounds for the worst-case metric distortion of any randomized election mechanism, and upper bounds for special cases, the problem of improving upper bounds for arbitrary metric spaces when $m \geq 4$ remains open. We conjecture that our lower bounds are optimal.

\begin{conjecture}\label{conj:UBs}
For each $m \geq 4$, there is a randomized election mechanism that guarantees distortion $2 + \gamma_m$ where $\gamma_4 \leq 0.04957$, $\gamma_5 \leq 0.06323$, up to $\displaystyle\lim_{m \to \infty} \gamma_m \leq 0.11264$.
\end{conjecture}

A less ambitious conjecture would be that there exists a positive constant $\eps$ such that there is a randomized election mechanism that guarantees distortion $3 - \eps$. 

Our linear programming approach suggests some possible ways of attempting to prove Conjecture~\ref{conj:UBs}, or a weaker version of it. We propose two new randomized voting mechanisms based on this approach, which we conjecture achieve the distortion claimed in Conjecture~\ref{conj:UBs}.

The first mechanism leverages Proposition~\ref{prop:suffbiased}, which tells us that in LP (\ref{eq:LP1}), it suffices to consider only the constraints that correspond to biased metrics. 

Suppose we have a biased metric $d$ corresponding to the vector $(x_1, x_2, ..., x_m)$, and without loss of generality suppose that $0 = x_1 \leq x_2 \leq \cdots \leq x_m$. Let $\Delta_i = x_{i + 1} - x_i$ The constraint corresponding to this metric is 
\begin{equation}\label{eq:biasedconstriant}
 \sum_{i \neq 1} (\SC(i) - \SC(1))p_i \leq \SC(1).
\end{equation}
By a similar computation to one in the proof of Lemma~\ref{lem:biasedtogen}, 
\begin{align*}
\SC(i) - \SC(1) &= x_1 s_{i \cg 1} + x_2 s_{1 \cg i \cg 2} + x_3s_{1,2 \cg i \cg 3} + \cdots + x_i s_{1,2,...,i-1 \cg i}\\
&= \Delta_1 s_{1\cg i} + \Delta_2 s_{1,2\cg i} + \cdots + \Delta_{i - 1}s_{1,2,...,i-1 \cg i}.
\end{align*}
$\SC(1)$ is somewhat more difficult to write down precisely, so we will settle for a lower bound. If $v \in S_{i \cg 1}$, $d(1, v) \geq  \frac12 x_i$. By picking the $i$ that makes this as large as possible for $v$, we get 
$$2\SC(1) \geq \sum_{i \neq 1} x_i s_{i \cg 0 \cg i + 1, ..., m} = \Delta_1(1 - s_{1 \cg 2,3,..,m}) + \Delta_2(1 - s_{1 \cg 3,4,...,m}) + \cdots + \Delta_{m - 1}(1 - s_{1 \cg m}).$$
This means that the constraint (\ref{eq:biasedconstriant}) is implied by the constraints
$$\sum_{j > k} s_{1,2,..., k \cg j}p_j \leq \frac12(1 - s_{1 \cg k + 1,k + 1,...,m})$$
for each $1 \leq k < m$. Considering all such constraints (and dropping the factor of $\frac12$ which can be absorbed into the variables, we can create the following LP.
\begin{align*}
\qquad\qquad\qquad\qquad\qquad\qquad\text{maximize} \quad &\sum_{i = 1}^m p_i & \label{eq:LP2} \tag{C}\\
\text{subject to}\quad  \sum_{j \notin I} s_{I \cg j} p_i &\leq  1 - \max_{i \in I} s_{i \cg I^c},  &\forall I \subseteq [m], I\neq \varnothing, [m]\\
p_i& \geq 0, &1 \leq i \leq m 
\end{align*}

Since the constraints in the above LP (after adding back the factor of $\frac12$) imply the constraints in (\ref{eq:LP1}), we have the following proposition.

\begin{proposition}\label{prop:lp2works}
For a given election, if the LP \textnormal{(\ref{eq:LP2})} has optimal objective at least $\beta$, then the optimal election mechanism achieves distortion at most $1 + \frac{2}{\beta}$.
\end{proposition}

We note that this proposition is enough to show us that the solution that assigns $p_i = \frac{\plu(i)}{1 - \plu(i)}$ (the \textsc{SmartDictatorship} mechanism considered by \cite{GHS20}) gets distortion $3 - 2\|\vect{{\plu}}\|^2 \leq 3 - 2/m$ since
$$\sum_{j \notin I} s_{I \cg j} \frac{\plu(j)}{1 - \plu(j)} \leq \sum_{j \notin I} \plu(j) \leq 1 - \sum_{i \in I} \plu(i)\leq  1 - \max_{i \in I} s_{i \cg I^c}$$
and we already showed that $\sum_i \frac{\plu(i)}{1 - \plu(i)} \geq \frac{1}{1 - \|\vect{{\plu}}\|^2}$.

Even though (\ref{eq:LP2}) is stronger than (\ref{eq:LP1}), we believe (based on some numerical evidence for small choices of $m$) that it could be sufficient to consider (\ref{eq:LP2}). In particular, we make the following conjecture.

\begin{conjecture}\label{conj:biased}
The election mechanism that solves the LP \textnormal{(\ref{eq:LP2})} and gives candidate $i$ weight $p_i/\sum_j p_j$ achieves the distortion demanded by Conjecture~\ref{conj:UBs}.
\end{conjecture}

By Proposition~\ref{prop:lp2works}, to prove this conjecture it suffices to show that \textnormal{(\ref{eq:LP2})} has sufficiently large optimal objective (to get distortion $3 - \Omega(1)$, we need to show it has optimal objective $1 + \Omega(1)$). 

Our second election mechanism solves the LP that optimizes over the $(0,1,2,3)$-metrics, given by (\ref{eq:LP0123}). This LP has the opposite issue of (\ref{eq:LP2}); while we know by Theorem~\ref{thm:0123UB} that (\ref{eq:LP0123}) has somewhat large optimal objective, we do not have that the constraints in this LP imply the constraints of the original LP (\ref{eq:LP1}). In fact, we have examples where some constraints of (\ref{eq:LP1}) are not implied by the constraints of (\ref{eq:LP0123}), but this just means that (\ref{eq:LP0123}) is not instance optimal. In these examples, (\ref{eq:LP0123}) still gives us a sufficiently good solution. 

\begin{conjecture}\label{conj:0123}
The election mechanism that solves the LP \textnormal{(\ref{eq:LP0123})} and gives candidate $i$ weight $p_i/\sum_j p_j$ achieves the distortion demanded by Conjecture~\ref{conj:UBs}.
\end{conjecture}

While Conjecture~\ref{conj:0123} is strictly harder to prove than Conjecture~\ref{conj:biased}, it is interesting because it would imply that the randomized metric distortion problem can be solved with extremely limited information; just the comparisons matrix $M$ and plurality vector $\vect{{\plu}}$. For reference, an election is specified with a list of $m!$ real numbers (the proportions of each ranking permutation), but if Conjecture~\ref{conj:0123} is true, we can solve the problem with just $m^2$ numbers. Another reason Conjecture~\ref{conj:0123} would be surprising is because of known lower bounds. \cite[Theorem~4]{GKM17} showed that a randomized mechanism with access only to $M$ cannot achieve distortion better than $3$, and \cite[Theorem~4]{GHS20} showed that a randomized mechanism  with access only to $\vect{{\plu}}$ cannot get distortion better than $3 - 2/m$. No lower bound (better than Theorem~\ref{thm:LBs}) is known for mechanisms with access to \emph{both} $M$ and $\vect{{\plu}}$.


\bibliography{pras}
\bibliographystyle{alpha}

\appendix
\section{Lower bound for three candidates}\label{sec:appendix}

In this section, we will go through some of the analysis in Section~\ref{sec:LBs} for the special case of $m = 3$, to more concretely see how the $(0,1,2,3)$-metrics work. For this special case, we will suspend our LP view of the distortion, and work directly with the definition instead.

Recall that in our example, we have three candidates $\{1,2,3\}$ and $47.3\%$ of the voters rank $(1, 3, 2)$, $42.4\%$ rank $(2, 3, 1)$, and $10.3\%$ rank $(3,2,1)$. Our three $(0,1,2,3)$-metrics are given in Figure~\ref{fig:0123metrics}.

\begin{figure}[h!]
\centering
\begin{tikzpicture}[scale=1.5]
\node[above] at (1, 2.5) {Metric $d_1$:};
\coordinate (A) at (0,2);
\coordinate (B) at (0,1);
\coordinate (C) at (0,0);

\coordinate (X) at (2,2);
\coordinate (Y) at (2,1);
\coordinate (Z) at (2,0);

\draw[thick, color=green] (A) -- (X);
\draw[thick, color=blue] (A) -- (Y);
\draw[thick, color=blue] (A) -- (Z);

\draw[thick, color=orange] (B) -- (X);
\draw[thick, color=blue] (B) -- (Y);
\draw[thick, color=blue] (B) -- (Z);

\draw[thick, color=orange] (C) -- (X);
\draw[thick, color=blue] (C) -- (Y);
\draw[thick, color=blue] (C) -- (Z);
\draw[color=black, fill=black] (A) circle (1pt);
\node[left] at (A) {1};
\draw[color=black, fill=black] (B) circle (1pt);
\node[left] at (B) {2};
\draw[color=black, fill=black] (C) circle (1pt);
\node[left] at (C) {3};
\draw[color=black, fill=black] (X) circle (1pt);
\node[right] at (X) {132};
\draw[color=black, fill=black] (Y) circle (1pt);
\node[right] at (Y) {231};
\draw[color=black, fill=black] (Z) circle (1pt);
\node[right] at (Z) {321};
\end{tikzpicture}  \qquad \qquad
\begin{tikzpicture}[scale=1.5]
\node[above] at (1, 2.5) {Metric $d_2$:};
\coordinate (A) at (0,2);
\coordinate (B) at (0,1);
\coordinate (C) at (0,0);

\coordinate (X) at (2,2);
\coordinate (Y) at (2,1);
\coordinate (Z) at (2,0);

\draw[thick, color=blue] (A) -- (X);
\draw[thick, color=orange] (A) -- (Y);
\draw[thick, color=red] (A) -- (Z);

\draw[thick, color=blue] (B) -- (X);
\draw[thick, color=green] (B) -- (Y);
\draw[thick, color=blue] (B) -- (Z);

\draw[thick, color=blue] (C) -- (X);
\draw[thick, color=orange] (C) -- (Y);
\draw[thick, color=blue] (C) -- (Z);
\draw[color=black, fill=black] (A) circle (1pt);
\node[left] at (A) {1};
\draw[color=black, fill=black] (B) circle (1pt);
\node[left] at (B) {2};
\draw[color=black, fill=black] (C) circle (1pt);
\node[left] at (C) {3};
\draw[color=black, fill=black] (X) circle (1pt);
\node[right] at (X) {132};
\draw[color=black, fill=black] (Y) circle (1pt);
\node[right] at (Y) {231};
\draw[color=black, fill=black] (Z) circle (1pt);
\node[right] at (Z) {321};
\end{tikzpicture}\qquad \qquad
\begin{tikzpicture}[scale=1.5]
\node[above] at (1, 2.5) {Metric $d_3$:};
\coordinate (A) at (0,2);
\coordinate (B) at (0,1);
\coordinate (C) at (0,0);

\coordinate (X) at (2,2);
\coordinate (Y) at (2,1);
\coordinate (Z) at (2,0);

\draw[thick, color=blue] (A) -- (X);
\draw[thick, color=red] (A) -- (Y);
\draw[thick, color=orange] (A) -- (Z);

\draw[thick, color=red] (B) -- (X);
\draw[thick, color=blue] (B) -- (Y);
\draw[thick, color=orange] (B) -- (Z);

\draw[thick, color=blue] (C) -- (X);
\draw[thick, color=blue] (C) -- (Y);
\draw[thick, color=green] (C) -- (Z);
\draw[color=black, fill=black] (A) circle (1pt);
\node[left] at (A) {1};
\draw[color=black, fill=black] (B) circle (1pt);
\node[left] at (B) {2};
\draw[color=black, fill=black] (C) circle (1pt);
\node[left] at (C) {3};
\draw[color=black, fill=black] (X) circle (1pt);
\node[right] at (X) {132};
\draw[color=black, fill=black] (Y) circle (1pt);
\node[right] at (Y) {231};
\draw[color=black, fill=black] (Z) circle (1pt);
\node[right] at (Z) {321};
\end{tikzpicture}
\caption{$(0,1,2,3)$-metrics for the given election. The green edges are 0, blue are 1, orange are 2, and red are 3.}\label{fig:0123metrics}
\end{figure}
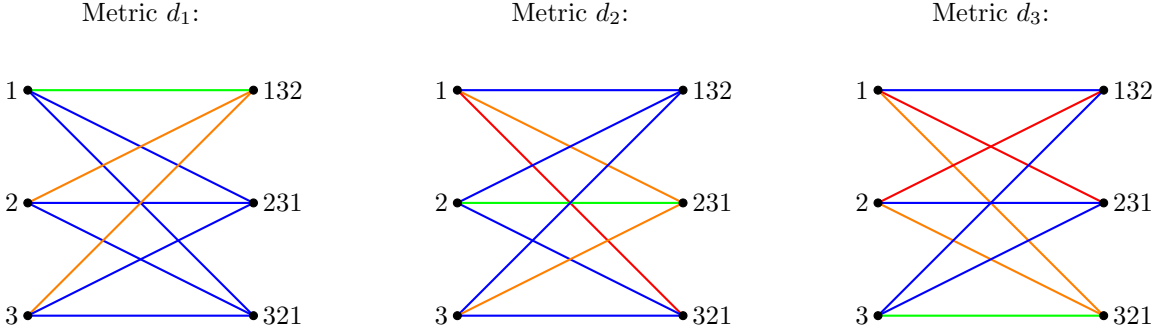
The social costs of the candidates in each of the three metrics are given by the following table.

\begin{center}
\begin{tabular}{c|c|c|c|}
        & $\SC(1)$ & $\SC(2)$ & $\SC(3)$\\
        \hline
$d_1$   & 0.527 & 1.473 & 1.473 \\
$d_2$   & 1.63  & 0.576 & 1.424  \\
$d_3$   & 1.951 & 2.049 & 0.897   
\end{tabular}
\end{center}
Now suppose we have an election rule that assigns weight $p_i$ to candidate $i$, and guarantees distortion $D$ (note that $p_1 + p_2 + p_3 = 1$). Then we can write out the three inequalities $\frac{\sum_i \SC(i)p_i}{\min_j \SC(j)} \leq D$ as
$$p_1 + 2.795p_2 + 2.795p_3 \leq D, \qquad\qquad 2.830p_1 + p_2 + 2.472p_3 \leq D \qquad\qquad 2.175p_1 + 2.284p_2 + p_3 \leq D.$$
Here we round to 3 decimal places for ease. Multiplying the first equation by 0.1514, the second by 0.1594, and the third by 0.1827 we get
\begin{align*}
0.1514 p_1 + 0.4232 p_2 + 0.4232 p_3 &\leq 0.1514 D\\
0.4511 p_1 + 0.1594 p_2 + 0.3940 p_3 &\leq 0.1594 D\\
0.3974 p_1 + 0.4173 p_2 + 0.1827 p_3 &\leq 0.1827 D
\end{align*}
Adding these up, (and being forgiving of numerical errors by allowing 0.999 = 1), we get 
$$p_1 + p_2 + p_3 \leq 0.4935 D\implies D \geq \frac{1}{0.4935} = 2.026.$$

\end{document}